%% file: SPAWC_2024_arxiv.tex
\documentclass[conference,10pt]{IEEEtran}
\IEEEoverridecommandlockouts

\normalsize
\usepackage{bbm}
\usepackage{adjustbox}
\usepackage{enumitem}
\usepackage{cite,algorithm,algorithmic,amsmath,amssymb,amsthm,empheq,mhsetup}
\usepackage{subfigure,amsfonts,balance}
\usepackage{epstopdf}
\usepackage{enumitem}
\usepackage{setspace}
\usepackage[dvipsnames]{xcolor}
\usepackage{amsmath}
\usepackage{array}
\usepackage{amsfonts,amssymb}
\usepackage{bm}
\usepackage{multirow}
\usepackage{tikz}
\usetikzlibrary{decorations.pathreplacing,calligraphy}
\usepackage{pgfplots}
\usepackage[left=0.63in, right=0.63in, top=0.62in, bottom=0.95in,centering]{geometry}

\DeclareMathOperator{\tr}{\mathrm{tr}}

\DeclareMathOperator{\EEE}{\mathbb{E}}

\DeclareMathOperator{\cov}{\mathrm{Cov}}

\DeclareMathOperator{\C}{\mathbb{C}}
\DeclareMathOperator{\0}{\mathbf{0}} 
\DeclareMathOperator{\F}{\mathbf{F}} 

\DeclareMathOperator{\aaa}{\mathbf{a}}

\DeclareMathOperator{\K}{\mathcal{K}}

\DeclareMathOperator{\E}{\mathbf{E}}
\DeclareMathOperator{\V}{\mathbf{V}}

\DeclareMathOperator{\h}{\mathbf{h}}
\DeclareMathOperator{\HH}{\mathbf{H}}

\DeclareMathOperator{\R}{\mathcal{R}}

\DeclareMathOperator{\G}{\mathbf{G}}
\DeclareMathOperator{\D}{\mathbf{D}}

\DeclareMathOperator{\CN}{\mathcal{CN}}
\DeclareMathOperator{\A}{\mathbf{A}}

\DeclareMathOperator{\B}{\mathbf{B}}

\DeclareMathOperator{\N}{\mathbf{N}}
\DeclareMathOperator{\NN}{\mathcal{N}}

\DeclareMathOperator{\e}{\mathbf{e}}

\DeclareMathOperator{\W}{\mathbf{W}}

\DeclareMathOperator{\II}{\mathbf{I}}

\DeclareMathOperator{\rr}{\mathbf{r}}
\DeclareMathOperator{\x}{\mathbf{x}}

\DeclareMathOperator{\y}{\mathbf{y}}
\DeclareMathOperator{\s}{\mathbf{s}}

\DeclareMathOperator{\n}{\mathbf{n}}

\DeclareMathOperator{\Q}{\mathbf{Q}}
\DeclareMathOperator{\q}{\mathbf{q}}
\DeclareMathOperator{\g}{\mathbf{g}}

\DeclareMathOperator{\Y}{\mathbf{Y}}

\DeclareMathOperator{\pphi}{\boldsymbol{\phi}}

\DeclareMathOperator{\VARPHI}{\boldsymbol{\varphi}}

\DeclareMathOperator{\Cov}{\mathrm{Cov}}

\makeatletter
\newsavebox\myboxA
\newsavebox\myboxB
\newlength\mylenA
\newcommand*\mybar[2][0.66]{%
	\sbox{\myboxA}{$\m@th#2$}%
	\setbox\myboxB\null% Phantom box
	\ht\myboxB=\ht\myboxA%
	\dp\myboxB=\dp\myboxA%
	\wd\myboxB=#1\wd\myboxA% Scale phantom
	\sbox\myboxB{$\m@th\overline{\copy\myboxB}$}%  Overlined phantom
	\setlength\mylenA{\the\wd\myboxA}%   calc width diff
	\addtolength\mylenA{-\the\wd\myboxB}%
	\ifdim\wd\myboxB<\wd\myboxA%
	\rlap{\hskip 0.5\mylenA\usebox\myboxB}{\usebox\myboxA}%
	\else
	\hskip -0.5\mylenA\rlap{\usebox\myboxA}{\hskip 0.5\mylenA\usebox\myboxB}%
	\fi}
\makeatother

\usepackage{tkz-euclide}
\tikzset{
bs/.pic = {                                      
    \draw[line width = 1pt,-round cap] (0,0.4\R) -- (-0.2\R,-0.2\R);
    \draw[line width = 1pt,-round cap] (0,0.4\R) -- (0.2\R,-0.2\R);
    \draw[line width = 1pt] (-0.2\R,-0.2\R) -- (0.133\R,0);
    \draw[line width = 1pt] (0.2\R,-0.2\R) -- (-0.133\R,0);
    \draw[line width = 1pt,-round cap] (-0.133\R,0) -- (0.067\R,0.2\R);
    \draw[line width = 1pt,-round cap] (0.133\R,0) -- (-0.067\R,0.2\R);
    \draw[line width = 1pt,-round cap] (-0.213\R,0.4\R) -- (0.213\R,0.4\R);
    \draw[line width = 1pt,-round cap] \foreach \x in {-0.21, -0.14,...,0.22} {(\x\R,0.4\R) -- (\x\R,0.47\R)};
    \node (dim) at (0,0)  [align=center,minimum width=0.5\R,minimum height=1\R] {};
    },
}

\tikzset{
 user/.pic = {     
    \draw[rounded corners=0.02\R] (-0.09\R,-0.17\R) rectangle (0.09\R,0.17\R) ;                     
    \draw[fill=gray] (-0.08\R,-0.12\R) rectangle (0.08\R,0.12\R);                          
    \draw[rounded corners=0.005\R] (-0.04\R,0.14\R) rectangle (0.04\R,0.15\R);                     
    \draw (0,-0.14\R) circle (0.012\R) ; 
    \node (dim) at (0,0)  [align=center,minimum width=0.2\R,minimum height=0.3\R] {};
    }
}

\newtheorem{theorem}{Theorem}

\newtheorem{lemma}{Lemma}
\newtheorem{remark}{Remark}

\allowdisplaybreaks

\begin{document}
\bstctlcite{IEEEexample:BSTcontrol}
\fontsize{10}{12}\rm

\title{
User-to-User Interference Mitigation in Dynamic TDD MIMO Systems with Multi-Antenna Users
 \vspace{-2mm}
 \thanks{This work was partially supported by the Wallenberg AI, Autonomous Systems and Software Program (WASP) funded by the Knut and Alice Wallenberg Foundation, and partially supported by ELLIIT.}
 \thanks{Tung T. Vu was with the Department of Electrical Engineering (ISY), Link\"{o}ping University, 581 83 Link\"{o}ping, Sweden.
}
 }

\author{
\IEEEauthorblockN{
Martin Andersson\IEEEauthorrefmark{1},
Tung T. Vu\IEEEauthorrefmark{2},
Pål Frenger\IEEEauthorrefmark{3},
Erik G. Larsson\IEEEauthorrefmark{1}
}
\IEEEauthorblockA{\small\IEEEauthorrefmark{1}Department of Electrical Engineering (ISY), Link\"{o}ping University, 581 83 Link\"{o}ping, Sweden}
\IEEEauthorblockA{\small\IEEEauthorrefmark{2}School of Engineering, Macquarie University, Australia}
\IEEEauthorblockA{\small\IEEEauthorrefmark{3}Ericsson Research, 583 30 Link\"{o}ping, Sweden}

Email: martin.b.andersson@liu.se, thanhtung.vu@mq.edu.au, 
pal.frenger@ericsson.com,
erik.g.larsson@liu.se
\vspace{-4mm}
}

\begin{figure*}[t!]
    \vspace{-20cm}
    \textcopyright~2024 IEEE. Personal use of this material is permitted. Permission from IEEE must be obtained for all other uses, in any current or future media, including reprinting/republishing this material for advertising or promotional purposes, creating new collective works, for resale or redistribution to servers or lists, or reuse of any copyrighted component of this work in other works.
\end{figure*}

\maketitle
\allowdisplaybreaks
\vspace{-0mm}
\begin{spacing}{1}
\begin{abstract}
We propose a novel method for user-to-user interference (UUI) mitigation in dynamic time-division duplex multiple-input multiple-output communication systems with multi-antenna users. Specifically, we consider the downlink data transmission in the presence of UUI caused by a user that simultaneously transmits in uplink. Our method introduces an overhead for estimation of the user-to-user channels by transmitting pilots from the uplink user to the downlink users. Each downlink user obtains a channel estimate that is used to design a combining matrix for UUI mitigation. We analytically derive an achievable spectral efficiency for the downlink transmission in the presence of UUI with our mitigation technique. Through numerical simulations, we show that our method can significantly improve the spectral efficiency performance in cases of heavy UUI.
%, but that it must be applied with caution to not deteriorate performance in low UUI scenarios.
\end{abstract}
\end{spacing}


\vspace{-0mm}
\section{Introduction}
\vspace{-0mm}
\label{sec:Introd}

Dynamic time-division duplex (DTDD) is a promising duplexing scheme for future multiple-input multiple-output (MIMO) communication systems \cite{Chowdhury2022TC, kim22VTC}. In DTDD, the time allocation for uplink (UL) and downlink (DL) transmissions, respectively, can be adapted dynamically in each coherence interval. Thus, the time allocation in DTDD can be optimized for the current network demands continuously, as opposed to current static TDD schemes where it is fixed. Moreover, DTDD lets each access point (AP) in the network choose its own UL and DL allocation, which can further increase data rates and decrease latencies. A special case of DTDD, which we consider in this paper, assigns each AP to either UL or DL throughout each coherence interval. In this way, a network-assisted (virtual) full-duplex network is realized without expensive full-duplex hardware \cite{fukue22A,Mohammadu2022SPAWC, Mohammadi2023JSAC}.

Although the potential of DTDD has been proven, there are still issues to overcome to reach it. Mainly, the cross-link interference (CLI) that is introduced in DTDD due to the simultaneous transmission in UL and DL, can limit the performance significantly \cite{KimCST}. Hence, CLI mitigation is a vital problem to solve. Most current studies on CLI mitigation for DTDD focus on the AP-to-AP interference \cite{Andersson2024ICASSP, deOlivindo2018WCL, daSilva2021WC}.
In contrast, the research on UUI mitigation is limited.

Furthermore, deploying multi-antenna users has shown the potential to significantly increase spectral efficiencies (SEs) in MIMO systems \cite{Li2016ICT, Mai2018GC, Sutton2021TVT}, and are thus likely to remain in future communication technologies. In this work, we aim to combine the advantages of multi-antenna users and DTDD. Specifically, we exploit the multiple user antennas to develop beamforming techniques for UUI mitigation. 

\textit{Paper contributions:} We aim to answer the question: \textit{"When is the DL SE improved by mitigating UUI in dynamic TDD MIMO systems with multi-antenna users?"} To this end, we consider a MIMO network with two APs, one operating in DL and one in UL. The DL AP transmits data to DL users in the presence of a UL user that simultaneously transmits to the UL AP, which causes UUI. We aim to enhance the SEs for the DL users by mitigating the UUI. The main contributions of this work are summarized as follows:
\begin{itemize}
    \item We propose a DTDD scheme that allocates a fraction of each coherence interval for user-to-user channel estimation. Importantly, our scheme does not increase the total required channel estimation overhead.
    \item We use the obtained channel estimates to design a combining matrix for each DL user, and apply this combiner to mitigate UUI.
    \item We derive an SE expression for the DL transmission in the presence of UUI that is achievable with our UUI mitigation technique.
    \item Numerical results show that our proposed method has the potential to remarkably improve the DL SEs in scenarios with vast UUI.
\end{itemize}

\section{System Model}
\label{sec:SystemModel}

We consider a MIMO communication system under a DTDD operation, where two $M$-antenna APs, one operating in DL and one in UL, serve users. The DL AP transmits signals to $K$ DL users, each of which equipped with $N_\text{d}$ antennas. Simultaneously, the UL AP receives the signal transmitted by a UL user equipped with $N_\text{u}$ antennas. The UL user unintentionally causes UUI to the DL users. The system model is visualized in Fig. \ref{Fig:system}.

\begin{remark}
    We emphasize that this system model is more general than the explicit system scenario herein suggests. The theory can cover other scenarios with minor modifications, such as the following. The antennas of the DL and UL APs do not have to be co-located, the $M$ antennas can be  
    distributed arbitrarily and connected through a backhaul network. Also, the $N_{\text{u}}$ UL antennas can be divided among several UL users.
\end{remark}

\begin{figure}[!t]
    \centering
    \begin{tikzpicture}
        \newdimen\R
        \R=1.75cm
        \draw (0.7,-2.3) pic (ulue) {user} node [xshift=-6mm, yshift=8mm] {} ;
        \draw (4,-0.5) pic (dlap) {bs} ;
        \draw (4,-2.3) pic (ulap) {bs};
        \draw (0,0) pic (dlue1) {user} node [yshift=-5mm]{$1$} ;
        \draw (0.7,-0.1) pic (dlue2) {user} node [yshift=-5mm]{$k$} ;
        \draw (2,0) pic (dlue3) {user}  node [yshift=-5mm]{$K$} ;
        \draw[->,black,dashed,thick,shorten >= 0.2cm] (0.6,-1.9) to (0,-0.5) ;
        \draw[->,black,dashed,thick,shorten >= 0.2cm] (0.7,-1.9) to (0.7, -0.6) ;
        \draw[->,black,dashed,thick,shorten >= 0.2cm] (0.8,-1.9) to (2, -0.5) ;
        \draw[<-,black,dashed,thick,shorten >= 0.2cm] (dlapdim.south) to (4,-0.65) ;
        \draw[<-,thick,shorten >= 0.2cm] (ulapdim.west) to node []{}(uluedim.east) ;
        
        \draw[thick, red] (0.7,-1.4) ellipse (0.75cm and 0.32cm) node [xshift=-17mm,yshift=1mm]{\begin{tabular}{c} $\HH_1,\dots,\HH_K$ \\ $\in \C^{N_\text{d} \times N_\text{u}}$ \end{tabular}};
        \draw[thick, red] (0.7,-1.4) ellipse (0.75cm and 0.32cm) node [xshift=12mm]{UUI};
        \draw[thick, blue, rotate=-30] (1.5,2) ellipse (1cm and 0.2cm) node [xshift=-27mm, yshift=4mm]{$\G_1,\dots,\G_K \in \C^{M \times N_\text{d}}$};
        
        \draw[black,thick,->] (dlapdim.north) to [out=135,in=45] node [align=center,midway,xshift=0mm,yshift=0mm]{} (0,0.4) ;
        \draw[black,thick,->] (dlapdim.north) to [out=135,in=45] node [align=center,midway,xshift=0mm,yshift=0mm]{} (0.7,0.3) ;
        \draw[black,thick,->] (dlapdim.north) to [out=135,in=45] node [align=center,midway,xshift=0mm,yshift=0mm]{} (2,0.4) ;
        
        \node  [align=center,left=0.01mm of uluedim.west] {UL user: \\ $N_\text{u}$ antennas};
        \node  [align=center,right=0.01mm of dlapdim.east] {DL AP: \\ $M$ antennas};
        \node  [align=center,left=0.05mm of dlue1dim.west] {DL users: \\ $N_\text{d}$ antennas};
        \node  [align=center,right=0.05mm of ulapdim.east] {UL AP: \\ $M$ antennas};         
        \node at (1.2,-0.08)[circle,fill,inner sep=0.8pt]{};
        \node at (1.35,-0.06)[circle,fill,inner sep=0.8pt]{};
        \node at (1.5,-0.04)[circle,fill,inner sep=0.8pt]{};
    \end{tikzpicture}
    \caption{Illustration of the considered DTDD MIMO system. Solid lines represent intended data signals and dashed lines are interference signals.}
    \label{Fig:system}
\end{figure}
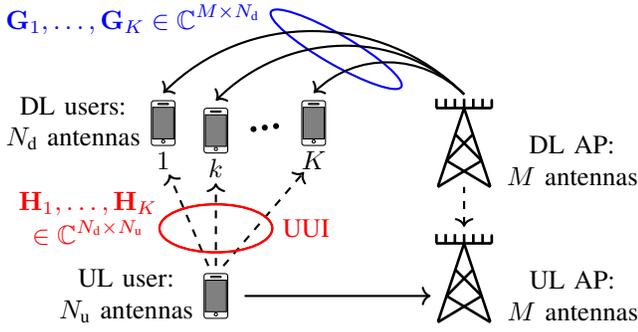

Let $\g_{kn} \in \C^{M \times 1},\ k \in \K \triangleq \{1, \dots, K\},\ n \in \NN_{\text{d}} \triangleq \{1, \dots, N_{\text{d}} \}$ denote the channel vector from the DL AP to antenna $n$ of DL user $k$. Similarly, define the channel vector from antenna $n$ of the UL user to DL user $k$ as $\h_{kn} \in \C^{N_{\text{d}} \times 1},\ k \in \K,\ n \in \NN_{\text{u}}\triangleq \{1,\dots,N_{\text{u}}\}$.
For each DL user $k$, let the collective channel matrices be
\begin{align}
    \G_k & = 
    \begin{bmatrix}
        \g_{k1} & \dots & \g_{kN_{\text{d}}}
    \end{bmatrix}
    \in \C^{M \times N_\text{d}},\ k \in \K
    \\
    \HH_k & =
    \begin{bmatrix}
        \h_{k1} & \dots & \h_{kN_{\text{u}}}
    \end{bmatrix}
    \in \C^{N_\text{d} \times N_\text{u}},\ k \in \K.
\end{align}

The uncorrelated Rayleigh fading model is considered for all channels, i.e., 
\begin{align}
    \G_k & = \sqrt{\alpha_k} \tilde{\G}_k, \ k \in \K
    \\
    \HH_k & = \sqrt{\beta_k} \tilde{\HH}_k, \ k \in \K,
\end{align}
where $\alpha_k$ and $\beta_k$ are the large-scale fading coefficients from the DL AP and the UL user, respectively, to DL user $k$. The matrices $\tilde{\G}_k \in \C^{M \times N_\text{d}}$ and $\tilde{\HH}_k \in \C^{N_\text{d} \times N_\text{u}}$ represent the small-scale fading and consist of independent and identically distributed (i.i.d.) $\CN(0,1)$ entries. All channels stay constant throughout one coherence interval of $\tau_{\text{c}}$ samples. 

\section{Proposed DTDD Scheme for UUI Mitigation}

We propose a DTDD scheme for UUI mitigation including the following steps in each coherence interval:

\begin{itemize}
    \item \textit{UL channel estimation}: During the first $\tau_{\text{p1}}$ samples, all $K$ DL users transmit UL pilots to estimate the user-to-AP channels $\{ \G_k \}$.
    \item \textit{UUI channel estimation}: The subsequent $\tau_{\text{p2}}$ samples are allocated to estimation of the UUI channels $\{ \HH_k \}$. The UL user transmits pilots to all $K$ DL users. These pilots are also received by the UL AP and are used to estimate the channel between the UL AP and the UL user.\footnote{Note that we focus on the DL transmission in this work. We do not explicitly consider the channel from the UL user to the UL AP, although this channel must be estimated in practice.}
    \item \textit{Data transmission}: The remaining $\tau_{\text{d}} \triangleq \tau_{\text{c}} - \tau_{\text{p1}} - \tau_{\text{p2}}$ samples are used for UL and DL data transmission.
\end{itemize}
Importantly, this scheme does not increase the total channel estimation overhead, since the $\tau_{\text{p2}}$ pilot samples would anyway be necessary to estimate all user-to-AP channels.
The tasks of the UL and DL users and APs in this DTDD scheme are presented in Fig. \ref{Fig:DTDD}, and are discussed in detail hereafter.

\begin{figure}[t!]
    \centering
    \begin{tikzpicture}[scale=0.8]
        % UL user
        \draw[thick] (1,0) -- (10,0);
        \draw[thick] (1,-1) -- (10,-1);
        \draw[thick] (1,0) -- (1,-1);
        \draw[thick] (10,0) -- (10,-1);
        \draw[thick] (4.7,0) -- (4.7,-1);
        \draw[thick] (3,0) -- (3,-1);

        \node at (-0.1,-0.5) {UL user:};
        \node[align=center, font=\small] at (2,-0.5) { Idle };
        \node[align=center, font=\small] at (3.85,-0.5) { Pilot \\ Trans. };
        \node at (7.25,-0.5) {UL Data Transmission};

        % UL AP
        \draw[thick] (1,-2) -- (10,-2);
        \draw[thick] (1,-1) -- (1,-2);
        \draw[thick] (10,-1) -- (10,-2);
        \draw[thick] (4.7,-1) -- (4.7,-2);
        \draw[thick] (3,-1) -- (3,-2);

        \node at (-0.1,-1.5) {UL AP:};
        \node[align=center, font=\small] at (2,-1.5) {Idle};
        \node[align=center, font=\small] at (3.85,-1.5) {UL \\ Ch. Est.};
        \node at (7.25,-1.5) {UL Data Reception};

        % DL users
        \draw[thick] (1,-3) -- (10,-3);
        \draw[thick] (1,-2) -- (1,-3);
        \draw[thick] (10,-2) -- (10,-3);
        \draw[thick] (4.7,-2) -- (4.7,-3);
        \draw[thick] (3,-2) -- (3,-3);

        \node at (-0.1,-2.5) {DL users:};
        \node[align=center, font=\small] at (2,-2.5) { Pilot \\ Trans.};
        \node[align=center, font=\small] at (3.85,-2.5) {UUI \\ Ch. Est.};
        \node at (7.25,-2.5) {DL Data Reception};

        % DL AP
        \draw[thick] (1,-4) -- (10,-4);
        \draw[thick] (1,-3) -- (1,-4);
        \draw[thick] (10,-3) -- (10,-4);
        \draw[thick] (4.7,-3) -- (4.7,-4);
        \draw[thick] (3,-3) -- (3,-4);

        \node at (-0.1,-3.5) {DL AP:};
        \node[align=center, font=\small] at (2,-3.5) {UL Channel \\ Estimation};
        \node[align=center, font=\small] at (3.85,-3.5) {Idle};
        \node at (7.25,-3.5) {DL Data Transmission};

        % Braces
        \draw [decorate,
        decoration = {calligraphic brace, raise=3pt, amplitude=3pt}, thick] (1.05,0) --  (2.95,0) node[pos=0.5,above=5pt,black]{$\tau_{\text{p1}}$};
        \draw [decorate,
        decoration = {calligraphic brace, raise=3pt, amplitude=3pt}, thick] (3.05,0) --  (4.65,0) node[pos=0.5,above=5pt,black]{$\tau_{\text{p2}}$};
        \draw [decorate, decoration = {calligraphic brace, raise=3pt, amplitude=3pt}, thick] (4.75,0) --  (9.95,0) node[pos=0.5,above=5pt,black]{$\tau_{\text{d}}$};

        \draw [decorate, decoration = {calligraphic brace, raise=3pt, amplitude=3pt, mirror}, thick] (1.05,-4) -- (9.95,-4) node[pos=0.5,below=5pt,black]{$\tau_{\text{c}}$};

        \draw[thick,->] (1, -5) to (10, -5) node[above, xshift=-5mm]{Time} ;
    \end{tikzpicture}
    \caption{Summary of the different phases and tasks performed by the UL and DL users and APs in each coherence interval with our DTDD scheme.}
    \label{Fig:DTDD}
\end{figure}
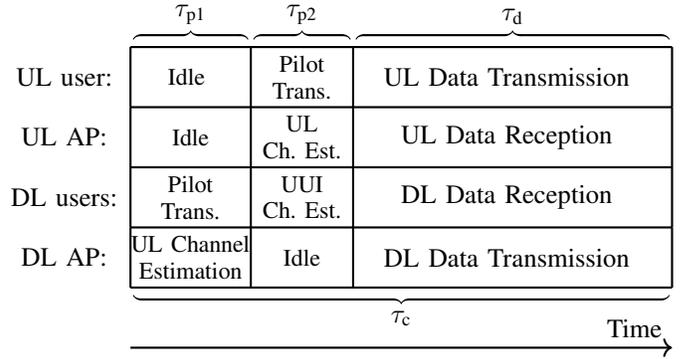

\subsection{Uplink Channel Estimation}

We use $\tau_{\text{p1}}$ symbols in each coherence interval for UL training by letting all DL users transmit pilot sequences to the DL AP. Let $\sqrt{\tau_{\text{p1}}} \pphi_{kn} \in \C^{\tau_{\text{p1}} \times 1}$ be the pilot sequence transmitted by antenna $n \in \NN_{\text{d}}$ of DL user $k \in \K$. We consider mutually orthogonal pilot sequences, i.e., $\pphi_{kn}^H \pphi_{k'n'} = 0$ if $(k,n) \neq (k',n')$ and $\pphi_{kn}^H \pphi_{kn} = 1$, where $\tau_{\text{p1}} \geq K N_{\text{d}}$. During the pilot transmission phase, the DL AP receives the $M \times \tau_{\text{p1}}$ signal $\Y_{\text{p1}}$, which can be expressed as
\begin{align}
    \Y_{\text{p1}} = \sqrt{\tau_{\text{p1}} \rho_{\text{u}}} \sum_{k \in \K} \sum_{n \in \NN_{\text{d}}} \g_{kn} \pphi_{kn}^H + \N_{\text{p1}},
\end{align}
where $\rho_{\text{u}}$ is the UL transmit power normalized by the noise power. The receiver noise $\N_{\text{p1}}$ consists of i.i.d. $\CN(0,1)$ elements. Then, de-spreading is applied by correlating $\Y_{\text{p1}}$ with each of the pilot sequences $\pphi_{kn}$ to obtain
\begin{align}
    \label{yp1}
    \nonumber
    \y_{\text{p1},kn} & \triangleq \Y_{\text{p1}} \pphi_{kn} 
    \\
    & = 
    \sqrt{\tau_{\text{p1}} \rho_{\text{u}}} \g_{kn} + \n_{\text{p1},kn},\ k \in \K,\ n \in \NN_{\text{d}},
\end{align}
where $\n_{\text{p1},kn} \sim \CN(\0_{M \times 1},\II_M)$. From \eqref{yp1}, the minimum mean square error (MMSE) estimate of the channel $\g_{kn}$ is \cite{ngo17TWC}
\begin{align}
    \hat\g_{kn} = \dfrac{\sqrt{\tau_{\text{p1}} \rho_{\text{u}}} \alpha_{k}}{1 + \tau_{\text{p1}} \rho_{\text{u}} \alpha_{k}} \y_{\text{p1},kn}, \ k \in \K,\ n \in \NN_{\text{d}}.
\end{align}
The collective channel estimate matrix for each DL user $k$ is:
\begin{align}
    \hat\G_k & = 
    \begin{bmatrix}
        \hat\g_{k1} & \dots & \hat\g_{kN_{\text{d}}}
    \end{bmatrix}
    \in \C^{M \times N_\text{d}},\ k \in \K.
\end{align}

\subsection{User-to-User Channel Estimation}

We follow the same steps for the DL users to estimate the UUI channels. Mutually orthogonal pilot sequences are used with $\tau_{\text{p2}} \geq N_{\text{u}}$. Let $\sqrt{\tau_{\text{p2}}} \VARPHI_{n} \in \C^{\tau_{\text{p2}} \times 1}$ be the pilot sequence transmitted by antenna $n \in \NN_{\text{u}}$ of the UL user, where $\VARPHI_n^H \VARPHI_{n'} = 0$ for $n \neq n'$ and $\VARPHI_n^H \VARPHI_n = 1$. The received $N_{\text{d}} \times \tau_{\text{p2}}$ pilot signal at DL user $k \in \K$ can be expressed as
\begin{align}
    \Y_{\text{p2},k} = \sqrt{\tau_{\text{p2}} \rho_{\text{u}}} \sum_{n \in \NN_{\text{u}}} \h_{kn} \VARPHI_{n}^H + \N_{\text{p2},k},
\end{align}
where the entries of $\N_{\text{p2},k}$ are i.i.d. $\CN(0,1)$. De-spreading is applied on $\Y_{\text{p2},k}$ to get
\begin{align}
    \nonumber
    \y_{\text{p2},kn} & \triangleq \Y_{\text{p2},k} \VARPHI_{n} 
    \\
    & = \sqrt{\tau_{\text{p2}} \rho_{\text{u}}} \h_{kn} + \n_{\text{p2},kn},\ k \in \K,\ n \in \NN_{\text{u}},
\end{align}
where $\n_{\text{p2},kn} \sim \CN(\0_{N_{\text{d}} \times 1}, \II_{N_{\text{d}}})$. The MMSE estimate of the UUI channel $\h_{kn}$ is \cite{ngo17TWC}
\begin{align}
    \hat\h_{kn} = \dfrac{\sqrt{\tau_{\text{p2}} \rho_{\text{u}}} \beta_{k}}{1 + \tau_{\text{p2}} \rho_{\text{u}} \beta_{k}} \y_{\text{p2},kn}, \ k \in \K,\ n \in \NN_{\text{u}}.
\end{align}
The collective channel estimate matrix for each user is then:
\begin{align}
    \hat\HH_k & = 
    \begin{bmatrix}
        \hat\h_{k1} & \dots & \hat\h_{kN_{\text{u}}}
    \end{bmatrix}
    \in \C^{N_{\text{d}} \times N_{\text{u}}},\ k \in \K.
\end{align}

\subsection{Downlink Data Transmission}
Let $\q_k \in \C^{N_{\text{d}} \times 1}$, be the vector of information symbols intended for DL user $k \in \K$, where $\q_k \sim \CN(\0_{N_{\text{d}} \times 1}, \II_{N_{\text{d}}})$ and all $\{ \q_k \}$ are mutually uncorrelated.
Define $\W_k \in \C^{M \times N_{\text{d}}}$ to be the DL precoding matrix 
% (including power control) 
for the same user. The signal transmitted from the DL AP is
\begin{align}
    \x = \sum_{k \in \K} \sqrt{\rho_{\text{d}}} \W_k \q_k \in \C^{M\times 1},
\end{align}
where $\rho_{\text{d}}$ denotes the normalized DL transmit power. Also, let $\sqrt{\rho_{\text{u}}} \s \in \C^{N_{\text{u}} \times 1}$ be the signal transmitted by the UL user, where $\EEE \{ \s \} = \0_{N_{\text{u}} \times 1}$ and all elements of $\s$ are mutually uncorrelated and have unit variance. The DL users receive the intended signals in the presence of UUI caused by the signal transmitted by the UL user. For DL user $k \in \K$, the received signal is
\begin{align}
    \y_{k} & \triangleq 
    \G_k^T \x + \sqrt{\rho_{\text{u}}} \HH_{k} \s + \n_{k} \in \C^{N_{\text{d}}\times 1},
\end{align}
where the receiver noise $\n_k$ consists of i.i.d. $\CN(0,1)$ elements. 

A combining matrix $\V_k \in \C^{N_{\text{d}} \times N_{\text{d}}}$ is applied for each DL user $k \in \K$ to obtain the following signal for data detection:
\begin{align}
    \nonumber
    \rr_k & \triangleq \V_k \y_k
    \\
    \nonumber
    & = \sqrt{\rho_{\text{d}}} \V_k \G_k^T \W_k \q_k 
    \\
    \nonumber
    & \quad + \sqrt{\rho_{\text{d}}} \sum_{k' \in \K \setminus \{ k \}} \V_k \G_k^T \W_{k'} \q_{k'}
    \\
    \nonumber
    & \quad + \sqrt{\rho_{\text{u}}} \V_k \HH_{k} \s + \V_k \n_{k}
    \\
    & = \F_k \q_k + \e_k,
\end{align}
where
\begin{align}
    \label{Fk}
    \F_k \triangleq \sqrt{\rho_{\text{d}}} \EEE \Big\{  \V_k \G_k^T \W_k \Big\}
\end{align}
is the average effective channel and
\begin{align}
    \nonumber
    \e_k & \triangleq \Bigg( \sqrt{\rho_{\text{d}}} \V_k \G_k^T \W_k - \EEE \Big\{ \sqrt{\rho_{\text{d}}} \V_k \G_k^T \W_k \Big\} \Bigg) \q_k
    \\
    \nonumber
    & \quad + \sqrt{\rho_{\text{d}}} \sum_{k' \in \K \setminus \{ k \}} \V_k \G_k^T \W_{k'} \q_{k'}
    \\
    & \quad + \sqrt{\rho_{\text{u}}} \V_k \HH_{k} \s + \V_k \n_{k}
\end{align}
is the effective noise. The mean and the covariance matrix of the effective noise term $\e_k$ are given in the following Lemma.
\begin{lemma}
    \label{lemma:effnoise}
    The effective noise $\e_k$ has zero mean, i.e., $\EEE \{ \e_k \} = \0_{N_{\text{d}} \times 1}$
    and its covariance matrix $\E_k \in \C^{N_{\text{d}} \times N_{\text{d}}}$ is given by
    \begin{align}
        \label{Ek}
        \nonumber
        \E_k & \triangleq \Cov \{ \e_k \}
        \\
        \nonumber
        & = \rho_{\text{d}} \EEE \Bigg\{ \V_k \G_k^T \left( \sum_{k' \in \K} \W_{k'} \W_{k'}^H \right) \G_k^* \V_k^H \Bigg\}
        \\
        \nonumber
        & \quad - \rho_{\text{d}} \EEE \{ \V_k \G_k^T \W_k \} \EEE \{ \W_k^H \G_k^* \V_k^H \}
        \\
        & \quad + \EEE \{ \V_k ( \rho_{\text{u}} \HH_k \HH_k^H + \II_{N_{\text{d}}} ) \V_k^H \}.
    \end{align}
\end{lemma}
\begin{proof}
    See Appendix \ref{Derivation_Ek}.
\end{proof}
The data symbols $\q_k$ are detected from $\rr_k$. An achievable SE for each DL user is presented in the following Theorem.
\begin{theorem}
    \label{theorem:DLSE}
    An achievable SE for DL user $k \in \K$ is
    \begin{align}
        \label{SE_dl}
        \mathtt{SE}_{k} = \dfrac{\tau_{\text{d}}}{\tau_{\text{c}}} \log_2 \left( \det \left( \II_{N_{\text{d}}} + \F_k^H \E_k^{-1} \F_k  \right) \right) \text{ bits/s/Hz},
    \end{align}
    where $\F_k$ and $\E_k$ are given by \eqref{Fk} and \eqref{Ek}, respectively.
\end{theorem}
\begin{proof}
     See Appendix \ref{proof:DLSE}. 
\end{proof}

\subsection{Combining Matrix Design for UUI Mitigation}
We propose to mitigate the UUI by using the UUI channel estimate to design the combining matrices $\V_k$ such that
\begin{align}
    \V_k \hat\HH_k \approx \0_{N_{\text{d}} \times N_{\text{u}}}.
\end{align}
This can be achieved with a projection matrix that projects onto the orthogonal complement of the dominant left singular vectors of $\hat\HH_k$. Specifically, take the economy-size singular value decomposition of $\hat\HH_k$ as $\hat\HH_k = \A_k \D_k \B_k^H,$ where $\D_k$ is a diagonal matrix with the non-zero singular values of $\hat\HH_k$ sorted in descending order on the diagonal, while $\A_k$ and $\B_k$ consist of the corresponding left and right singular vectors, respectively. Denote by $\aaa_{k,i} \in \C^{N_{\text{d}} \times 1}$ the $i$th left singular vector, i.e., the $i$th column of $\A_k$. Then, we construct the combining matrices such that they project onto the orthogonal complement of the $C$ dominant left singular vectors of $\hat\HH_k$, where $C \leq \min \{ N_{\text{d}}-1, N_{\text{u}} \}$:
\begin{align}
    \label{Vk_design}
    \V_k = \II_{N_{\text{d}}} - \sum_{i=1}^C \aaa_{k,i} \aaa_{k,i}^H.
\end{align}

\section{Baseline Scenarios}

To evaluate the performance of the proposed UUI mitigation scheme, the following two baselines are considered for comparison. The first baseline presents the scenario where perfect UUI mitigation is achieved, and the second baseline considers the scenario where the UUI is not mitigated. 

\subsection{Perfect UUI Mitigation (Genie)}

The received signal at DL user $k \in \K$ after subtraction of the UUI is given by
\begin{align}
    \bar \y_k \triangleq \y_k - \sqrt{\rho_{\text{u}}} \HH_k \s = \bar\F_k \q_k + \bar\e_k,
\end{align}
where
\begin{align}
    \label{barFk}
    \bar\F_k \triangleq \sqrt{\rho_{\text{d}}} \EEE \Big\{ \G_k^T \W_k \Big\},
\end{align}
and $\bar\e_k$ is the zero mean effective noise with covariance matrix
\begin{align}
    \label{barEk}
    \nonumber
    \bar\E_k & \triangleq \Cov \{ \bar\e_k \} 
    \\
    \nonumber
    & = \rho_{\text{d}} \EEE \Bigg\{ \G_k^T \left( \sum_{k' \in \K} \W_{k'} \W_{k'}^H
    \right) \G_k^* \Bigg\}
    \\
    & \quad - \rho_{\text{d}} \EEE \{ \G_k^T \W_k \} \EEE \{ \W_k^H \G_k^* \} + \II_{N_\text{d}}.
\end{align}
The desired symbols $\q_k$ are detected from $\bar\y_k$. An achievable SE is given in the following Theorem.

\begin{theorem}
    \label{theorem:DLSE_genie}
    An achievable SE for DL user $k \in \K$ with perfect UUI mitigation is
    \begin{align}
        \label{SE_dl_genie}
        \mathtt{SE}_{k}^{\mathtt{g}} = \dfrac{\tau_{\text{d}}}{\tau_{\text{c}}} \log_2 \left( \det \left( \II_{N_{\text{d}}} + \bar\F_k^H \bar\E_k^{-1} \bar\F_k \right) \right) \text{ bits/s/Hz},
    \end{align}
    where $\mybar\F_k$ and $\mybar\E_k$ are given by \eqref{barFk} and \eqref{barEk}, respectively.
\end{theorem}
\begin{proof}
    Similar to that of Theorem \ref{theorem:DLSE}, and hence, omitted.
\end{proof}

\subsection{No UUI Mitigation (Na\"ive)}

Since no UUI mitigation is performed, the desired symbols $\q_k$ are detected from $\y_k$. An achievable SE is given in the following Theorem.
\vspace{-2mm}
\begin{theorem}
    \label{theorem:DLSE_naive}
    An achievable SE for DL user $k \in \K$ without UUI mitigation is
    \begin{align}
        \label{SE_dl_naive}
        \mathtt{SE}_{k}^{\mathtt{n}} = \dfrac{\tau_{\text{d}}}{\tau_{\text{c}}} \log_2 \left( \det \left( \II_{N_{\text{d}}} + \bar\F_k^H \check\E_k^{-1} \bar\F_k \right) \right) \text{ bits/s/Hz},
    \end{align}
    where $\bar\F_k$ is given in \eqref{barFk} and
    \begin{align}
        \check\E_k & \triangleq \bar\E_k + \rho_{\text{u}} \EEE \{ \HH_{k} \HH_k^H \}.
    \end{align}
\end{theorem}
\begin{proof}
     Follows from that of Theorem \ref{theorem:DLSE} by letting $\V_k = \II_{N_{\text{d}}}$.
\end{proof}

\vspace{-0mm}
\section{Numerical results}
\label{sec:sims}

Consider a MIMO system with $M=8$, $K=1$, $N_{\text{d}} = 4$, $N_{\text{u}} = 2$, and $\tau_{\text{c}} = 200$. The pilot sequence lengths are chosen as $\tau_{\text{p1}} = K N_{\text{d}} = 4$ and $\tau_{\text{p2}} = N_{\text{u}} = 2$, respectively. We randomly distribute the users and the APs in a square of $250$ m $\times$ $250$ m, with wrapped around edges. The large-scale fading coefficients $\{ \alpha_k \}$ and $\{ \beta_k \}$ are calculated from the outdoor 3GPP Urban Microcell model \cite[Eqs. (37), (38)]{emil20TWC}. 

We consider combining matrices $\{ \V_k \}$ constructed with $C=1$ and $C=2$, respectively. Maximum-ratio precoding and equal power allocation is applied in the DL by letting $\W_k = \hat\G_k^* \Big/ \sqrt{ K \EEE \{ \tr ( \hat\G_k^* \hat\G_k^T ) \}}$.

Define the median signal-to-noise ratios (SNRs) observed at the receiver side when transmitting in UL, in DL, and over the UUI channels, as $\mathtt{SNR}_{\text{u}}$, $\mathtt{SNR}_{\text{d}}$, and $\mathtt{SNR}_{\text{uui}}$. We consider a scenario where the UUI SNR is relatively high, by choosing $\mathtt{SNR}_{\text{u}} = \mathtt{SNR}_{\text{d}} = \mathtt{SNR}_{\text{uui}} = 20$ dB.
In Fig. \ref{Fig:SE_CDF}, we present the cumulative distribution functions (CDFs) of the SEs achieved by the proposed method and the two baselines over 1000 channel realizations. It can be seen that the impact of UUI is major as the genie heavily outperforms the other methods. Our proposed method outperforms the na\"ive method in, roughly, $50\%$ of the cases for both $C=1$ and $C=2$. Furthermore, it is observed that the increase in the $80\%$-likely SE with our method, compared to the na\"ive method, is approximately doubled with $C=1$ and almost five-fold with $C=2$, respectively. However, the peak SEs are heavily deteriorated, especially for $C=2$. The reason is that when the proposed combining matrix is applied, it unintentionally cancels parts of the desired signal as well. Hence, when the UUI is small, the loss in the desired signal can typically not be compensated for by mitigating the UUI. 

\begin{figure}[!t]
    \centering
    \input{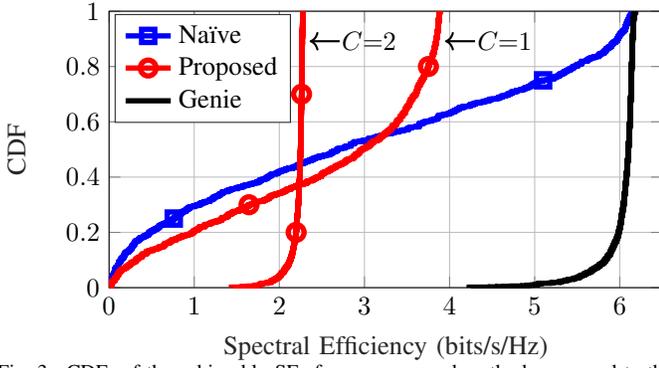}
    \vspace{-5mm}
    \caption{CDFs of the achievable SEs for our proposed method compared to the baseline scenarios without UUI mitigation and with perfect UUI mitigation.}
    \label{Fig:SE_CDF}
\end{figure}

\begin{figure}[!t]
    \centering
    \input{fig_CDF_rel}
    \vspace{-5mm}
    \caption{The $10\%$-likely, $50\%$-likely, and $90\%$-likely relative SEs of our method with $C=2$ compared to the na\"ive method, for varying UUI SNRs.}
    \label{Fig:CDF_rel}
\end{figure}
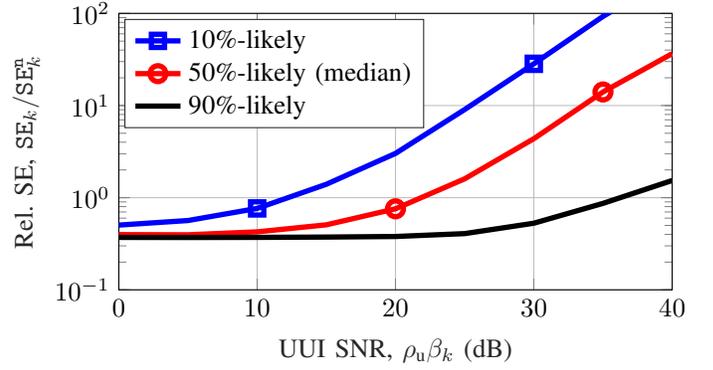

To further highlight the circumstances under which our method outperforms the na\"ive method, we now consider the relative performance $\mathtt{SE}_{k}/\mathtt{SE}_{k}^{\mathtt{n}}$. We restrict ourselves to $C=2$, and calculate the $10\%$-likely, $50\%$-likely (median), and $90\%$-likely relative SEs, for varying UUI SNRs $\rho_{\text{u}} \beta_k$. The results are presented in Fig. \ref{Fig:CDF_rel}, which shows that applying our method is more likely to improve the SE than not (i.e., the median exceeds $1$), when the UUI SNR is larger than approximately $20$ dB, which is in accordance with the results from Fig. \ref{Fig:SE_CDF}. It is also seen that our method will almost surely (i.e., in $90\%$ of cases) improve the SE when the UUI SNR exceeds roughly $35$ dB. Furthermore, in the same SNR region, the SE improvement is more than ten-fold in $50\%$ of cases. In contrast, we see that the SE is more than $90\%$ likely to be deteriorated by applying UUI mitigation when the UUI SNR falls below $10$ dB.

To summarize, our method for UUI mitigation can greatly improve DL SEs under heavy UUI. This leads to the question: ``What are the specific circumstances under which it is beneficial to apply our UUI mitigation method while avoiding deteriorated peak SEs?''. The answer to this open question is left for future work.

\section{Conclusions}

We propose a method for user-to-user interference mitigation in MIMO systems with DTDD operation. Our method introduces a phase for user-to-user channel estimation in each coherence interval, without increasing the overall resources spent on channel estimation. The obtained channel estimates are then used by each interfered DL user to design a combining matrix for UUI mitigation. We analytically derive an SE expression that is achievable with our proposed method. Numerical results show that our method can significantly improve the DL SEs in scenarios with severe UUI. However, it decreases SEs in other cases, and we emphasize the need for future work to investigate in detail when it is beneficial to apply UUI mitigation.
\appendix
\section{}

\subsection{Proof of Lemma \ref{lemma:effnoise}}
\label{Derivation_Ek}

\vspace{1mm}
Since all $\q_k$, $\s$, and $\n_k$ are zero mean and independent of all $\V_k$, $\G_k$, $\W_k$, and $\HH_k$, we immediately get that
\begin{align}
    \EEE \{ \e_k \} = \0_{N_{\text{d}} \times 1}.
\end{align}
Using the same properties and that all $\q_k$, $\s$, and $\n_k$ are mutually uncorrelated one can show that all terms in $\e_k$ are uncorrelated, and hence we have 
\begin{align}
    \nonumber
    \E_k & = \Cov \{ \e_k \}
    \\
    \nonumber
    \\
    \nonumber
    \\
    \nonumber
    & = \cov \Big\{ \Big( \sqrt{\rho_{\text{d}}} \V_k \G_k^T \W_k
    \\
    \label{BU}
    & \qquad \qquad \quad - \EEE \Big\{ \sqrt{\rho_{\text{d}}} \V_k \G_k^T \W_k \Big\} \Big) \q_k \Big\}
    \\
    \label{IUI}
    & \quad + \sum_{k' \in \K \setminus \{ k \}} \cov \Big\{ \sqrt{\rho_{\text{d}}} \V_k \G_k^T \W_{k'} \q_{k'} \Big\}
    \\
    \label{UUI}
    & \quad + \cov \{ \sqrt{\rho_{\text{u}}} \V_k \HH_{k} \s \} 
    \\
    \label{RN}
    & \quad + \cov \{ \V_k \n_{k} \}.
\end{align}

We now calculate each of these terms.

\subsubsection{Calculating \eqref{RN}}

We use the property
\begin{align}
    \Cov \{ \V_k \n_k \} & = \EEE \{ \Cov \{ \V_k \n_k | \V_k \} \}.
\end{align}
Note that, given $\V_k$, the conditional distribution of $\V_k \n_k$ is $\CN ( \0_{N_{\text{d}} \times 1}, \V_k \V_k^H )$ and hence
\begin{align}
    \label{Cov_Vn}
    \Cov \{ \V_k \n_k \} & = \EEE \{ \V_k \V_k^H \}.
\end{align}

\subsubsection{Calculating \eqref{UUI}}

In the same way, we get
\begin{align}
    \label{Cov_VHs}
    \nonumber \Cov \{ \sqrt{\rho_{\text{u}}} \V_k \HH_k \s \} & = \rho_{\text{u}} \EEE \{ \Cov \{ \V_k \HH_k \s | \V_k \HH_k \} \}
    \\
    & = \rho_{\text{u}} \EEE \{ \V_k \HH_{k} \HH_k^H \V_k^H \}
\end{align}

\subsubsection{Calculating \eqref{IUI}}
Again, a similar calculation gives
\begin{align}
    & \nonumber \Cov \{ \sqrt{\rho_{\text{d}}} \V_k \G_k^T \W_{k'} \q_{k'} \}
    \\
    \nonumber
    & = \rho_{\text{d}} \EEE \{ \Cov \{ \V_k \G_k^T \W_{k'} \q_{k'} | \V_k \G_k^T \W_{k'} \} \}
    \\
    & = \rho_{\text{d}} \EEE \{ \V_k \G_k^T \W_{k'} \W_{k'}^H \G_k^* \V_k^H \}.
\end{align}
Hence,
\begin{align}
    \label{Cov_VGWq}
    \nonumber
    & \sum_{k' \in \K \setminus \{ k \}} \cov \Big\{ \sqrt{\rho_{\text{d}}} \V_k \G_k^T \W_{k'} \q_{k'} \Big\}
    \\
    \nonumber
    & = \rho_{\text{d}} \sum_{k' \in \K \setminus \{ k \}} \EEE \{ \V_k \G_k^T \W_{k'} \W_{k'}^H \G_k^* \V_k^H \}
    \\
    & = \rho_{\text{d}} \EEE \Bigg\{ \V_k \G_k^T \left( \sum_{k' \in \K \setminus \{ k \}} \W_{k'} \W_{k'}^H \right) \G_k^* \V_k^H \Bigg\}.
\end{align}

\subsubsection{Calculating \eqref{BU}}

Define 
\begin{align}
    \label{Qk}
    \Q_k \triangleq \Big( \sqrt{\rho_{\text{d}}} \V_k \G_k^T \W_k - \EEE \Big\{ \sqrt{\rho_{\text{d}}} \V_k \G_k^T \W_k \Big\} \Big).
\end{align}
This gives
\begin{align}
    \label{CovBU}
    \nonumber
    \Cov \{ \Q_k \q_k \} & = \EEE \{ \Cov \{ \Q_k \q_k | \Q_k \} \}
    \\
    & = \EEE \{ \Q_k \Q_k^H \}.
\end{align}
Let $\bar\Q_k \triangleq \sqrt{\rho_{\text{d}}} \V_k \G_k^T \W_k$, such that $\Q_k = \bar\Q_k - \EEE \{ \bar\Q_k \}$. By substituting this into
\eqref{CovBU} we get
\begin{align}
    \label{Cov_Qq}
    \nonumber
    \EEE \{ \Q_k^H \Q_k \} & = \EEE \{ (\bar\Q_k - \EEE \{ \bar\Q_k \})(\bar\Q_k - \EEE \{ \bar\Q_k \})^H \}
    \\
    \nonumber
    & = \EEE \{ \bar\Q_k \bar\Q_k^H \} - \EEE \{ \bar\Q_k \} \EEE \{ \bar\Q_k \}^H
    \\
    \nonumber
    & = \rho_{\text{d}} ( \EEE \{ \V_k \G_k^T \W_k \W_k^H \G_k^* \V_k^H \} 
    \\
    & \qquad - \EEE \{ \V_k \G_k^T \W_k \} \EEE \{ \W_k^H \G_k^* \V_k^H \}).
\end{align}

\subsubsection{Collect all terms to obtain the final expression}
\begin{align}
    \nonumber
    \E_k & = \rho_{\text{d}} \EEE \Bigg\{ \V_k \G_k^T \left( \sum_{k' \in \K} \W_{k'} \W_{k'}^H \right) \G_k^* \V_k^H \Bigg\}
    \\
    \nonumber
    & \quad - \rho_{\text{d}} \EEE \{ \V_k \G_k^T \W_k \} \EEE \{ \W_k^H \G_k^* \V_k^H \}
    \\
    & \quad + \EEE \{ \V_k ( \rho_{\text{u}} \HH_k \HH_k^H + \II_{N_{\text{d}}} ) \V_k^H \}.
\end{align}
\vspace{-8mm}
\subsection{Proof of Theorem~\ref{theorem:DLSE}}
\label{proof:DLSE}
\vspace{-3mm}
From \cite[Theorem 2]{Li2016ICT}, an achievable SE is given by 
\begin{align}
    \label{achievablerate}
    \mathtt{SE}_{k} = \dfrac{\tau_{\text{d}}}{\tau_{\text{c}}} \log_2 \left( \det \left( \II_{N_\text{d}} + \F_k^H \mybar{\boldsymbol{\Xi}}_k \F_k \right) \right),
\end{align}
with $\F_k$ as in \eqref{Fk} and $\mybar{\boldsymbol{\Xi}}_k$ defined by
\begin{align}
    \label{Xibar}
    \mybar{\boldsymbol{\Xi}}_k \triangleq \left( \boldsymbol{\Xi}_k^{-1} - \F_k \F_k^H \right)^{-1},
\end{align}
where $\boldsymbol{\Xi}_k$ is chosen such that the linear MMSE (LMMSE) estimate $\hat\q_k$ of $\q_k$ given $\rr_k$ is
\begin{align}
    \hat\q_k = \F_k^H \boldsymbol{\Xi}_k \rr_k.
\end{align}
By using \cite[Equation (15.64)]{kay93} and $\q_k \sim \CN(\0_{N_{\text{d}} \times 1}, \II_{N_{\text{d}}})$, the LMMSE estimate is
\begin{align}
    \hat\q_k = \F_k^H \left( \F_k \F_k^H + \E_k \right)^{-1} \rr_k,
\end{align}
from which we identify
\begin{align}
    \label{Xi}
    \boldsymbol{\Xi}_k = \left( \F_k \F_k^H + \E_k \right)^{-1}.
\end{align}
By inserting \eqref{Xi} into \eqref{Xibar} we get
\begin{align}
    \mybar{\boldsymbol{\Xi}}_k = \E_k^{-1},
\end{align}
which after substitution into \eqref{achievablerate} completes the proof.

\ifCLASSOPTIONcaptionsoff
  \newpage
\fi

\begin{spacing}{1}
\bibliographystyle{IEEEtran}
\bibliography{IEEEabrv,refs}
\end{spacing}

\end{document}

%% file: fig_CDF_rel.tex
\begin{tikzpicture}

\begin{axis}[%
width=0.4\textwidth,
height=0.2\textwidth,
at={(0.25in,0.481in)},
scale only axis,
xmin=0,
xmax=40,
xlabel style={font=\color{white!15!black}},
xlabel={$\text{UUI SNR, $\rho_{\text{u}} \beta_k$  (dB)}$},
ymode=log,
ymin=0.1,
ymax=100,
yminorticks=true,
ylabel style={font=\color{white!15!black}},
ylabel={
Rel. SE, $\mathtt{SE}_{k}/\mathtt{SE}_{k}^{\mathtt{n}}$},
xmajorgrids,
ymajorgrids,
legend style={at={(0.01,0.98)}, anchor=north west, legend cell align=left, align=left, draw=white!15!black}
]

\addplot [color=blue, line width=2.0pt, mark size=2.5pt, mark=square, mark options={solid, blue}, mark repeat=4, mark phase=3]
  table[row sep=crcr]{%
0	0.504240807956891\\
5	0.56496471602795\\
10	0.764834773515024\\
15	1.39766139973656\\
20	3.02139545957951\\
25	9.04866491600269\\
30	28.3713715730019\\
35	92.7074160264615\\
40	256.256903250569\\
};
\addlegendentry{10\%-likely}

\addplot [color=red, line width=2.0pt, mark size=3.0pt, mark=o, mark options={solid, red}, mark repeat=3, mark phase=5]
  table[row sep=crcr]{%
0	0.3974038228917\\
5	0.39574342282878\\
10	0.425834055708564\\
15	0.507419579805679\\
20	0.757485650205043\\
25	1.6015576720619\\
30	4.36574491950099\\
35	14.0408310278499\\
40	36.5624696464854\\
};
\addlegendentry{50\%-likely (median)}

\addplot [color=black, line width=2.0pt]
  table[row sep=crcr]{%
0	0.370624774536914\\
5	0.369971382797592\\
10	0.370960410474641\\
15	0.373859157550448\\
20	0.37921558110828\\
25	0.40746894458566\\
30	0.528282071871685\\
35	0.86732442352561\\
40	1.53609881409505\\
};
\addlegendentry{90\%-likely}

\end{axis}

\end{tikzpicture}%